\documentclass [journal,onecolumn,11pt]{IEEEtran}
\usepackage{amsfonts,amsmath,amssymb}
\usepackage{indentfirst, setspace}
\usepackage{url,float}
\usepackage{color}
\usepackage[a4paper,ignoreall]{geometry}
\usepackage{longtable}
\usepackage{graphicx}
\usepackage[a4paper,ignoreall]{geometry}
\usepackage{multicol}
\usepackage{stfloats}
\usepackage{enumerate}
\usepackage{cite}
\usepackage{amsthm}
\usepackage{multirow}
\usepackage{amssymb}
\usepackage[square, comma, sort&compress, numbers]{natbib}

\def\qu#1 {\fbox {\footnote {\ }}\ \footnotetext { From Qu: {\color{red}#1}}}
\def\hqu#1 {}
\def\kq#1 {\fbox {\footnote {\ }}\ \footnotetext { From KangQuan: {\color{blue}#1}}}
\def\hkq#1 {}
\newcommand{\mqu}[1]{{{\color{red}#1}}}

\newtheorem{Th}{Theorem}[section]

\newtheorem{Lemma}[Th]{Lemma}
\newtheorem{Def}[Th]{Definition}

\newcommand{\gf}{{\mathbb F}}

\makeatletter
\newcommand{\figcaption}{\def\@captype{figure}\caption}
\newcommand{\tabcaption}{\def\@captype{table}\caption}
\makeatother

\begin{document}
	\title{New Classes of Permutation Binomials and Permutation Trinomials over Finite Fields}
	
\author{Kangquan~Li, Longjiang~Qu, and Xi~Chen
        \thanks{The authors are with the College of Science, National University of Defense Technology, ChangSha, 410073,
                China, email: 940672099@qq.com, ljqu\_happy@hotmail.com, 1138470214@qq.com.
                The research of this paper is supported by the NSFC of China
                under Grant  61272484, the National Basic
Research Program of China(Grant No. 2013CB338002), the Basic Research Fund of National University of Defense Technology (No. CJ 13-02-01)
and the Program for New Century Excellent Talents in University (NCET).}
               }
	
	\maketitle{}
	
	\begin{abstract}
	Permutation polynomials over finite fields play important roles in finite fields theory.
	They also have wide applications in many areas of science and engineering such as coding theory, cryptography,
    combinational  design, communication theory and so on. Permutation binomials and trinomials attract people's interest due to their simple algebraic form and additional extraordinary properties. In this paper, {several} new classes of permutation binomials and
     permutation trinomials are constructed. {Some of these permutation polynomials are generalizations of known ones. }
    \end{abstract}


    \section{Introduction}
    A polynomial $f \in \mathbb{F}_{q}[x]$ is called a permutation polynomial of $\mathbb{F}_{q}$
     if the associated polynomial function $f:c\to f(c)$ from $\mathbb{F}_{q}$ into $\mathbb{F}_{q}$ is a permutation of $\mathbb{F}_{q}$.
     Permutation polynomials over finite fields play important roles in finite fields theory.
     They also have wide applications in coding theory, cryptography, combinational design, communication theory.
      The study of permutation polynomial can date back to Hermite  \cite{Hermite} and Dickson  \cite{Dic}.
      {There are numerous books and survey papers on the subject covering different periods in the development of this active area
       \cite[Ch.18]{22},  \cite{XH},    \cite[Ch.7]{LN},  \cite[Ch.8]{64},  \cite{XH1}, while the survey by X. Hou  \cite{XH1} in 2015 is the most recent one.}

    Permutation binomials and trinomials attract people's interest due to their simple algebraic form and additional extraordinary properties. For instances,
   {a certain permutation trinomial of $\gf_{2^{2m+1}}$ was a major  ingredient in the proof of the Welch conjecture  \cite{HD}.}
     The discovery of other class of permutation trinomials by Ball and Zieve \cite{SM} provides a way to prove the construction of the Ree-Tits symplectic spreads of $\mathbf{PG}(3,q)$. For more relevant progresses, readers can consult a paper \cite{XH} presented by X.Hou, which conducted a recent survey on permutation binomials and trinomials in 2013.

    {Although quite a few permutation binomials and permutation trinomials have been found,  however,
an explicit and unify characterization of them is still missing and seems to be elusive today. %
  Therefore, it is both interesting and important to find more explicit classes of them and try to find the hidden reason for their permutation property.
In this paper, we  present several new classes of} permutation binomials over finite fields with even extensions and {permutation} trinomials over finite fields with even characteristics. Therefore {we will review} the known permutation binomials over $\mathbb{F}_{q^2}$ and permutation trinomials over $\mathbb{F}_{2^m}$ {in Section \ref{kpb} and Section \ref{kpt},
    respectively} for the convenience of the reader. {For other classes of permutation binomials or permutation trinomials, please refer to  \cite{XH}.}

In Section \ref{sec1}, we {introduce} the definition of the {multiplicative equivalence of permutation polynomials } and some lemmas. In Section \ref{sec2}, we {provide} new classes of permutation binomials through the Hermite's criterion \cite{Hermite} and considering the number of solutions of special equations. In Section \ref{sec3}, we give five new classes of permutation trinomials. {Section \ref{sec4} is the conclusion.} In the paper, ${\operatorname{Tr}}_{k}(x)$ is treated as the absolute trace function on $\mathbb{F}_{2^{k}}$. The algebraic closure of $\mathbb{F}_{q}$ is denoted {by} $\bar{\mathbb{{F}}}_{q}$. For integer $d>0$, $\mu_{d}=\{x \in \bar{\mathbb{{F}}}_{q}:x^{d}=1\}$. Others symbols follow the standard notation of finite fields.

    \section{Preparation}
    \label{sec1}
    {Let $f$ and $g$ be two polynomials in $\mathbb{F}_{q}[x]$ satisfying that $f(x)=g(x^d)$,
    where $1 \le d \le q-1$ is an integer such that $\mathrm{gcd}(d,q-1)=1$. Then  the number of the terms of $f$  is equal to that of $g$, and
     $f$ is a permutation polynomial if and only if so is $g$. Particularly, $f$ is a permutation binomial resp. permutation trinomial if and
     only if so is $g$. Hence we introduce the following definition of multiplicative equivalence.}
    \begin{Def}
    Two permutation polynomials $f(x)$ and $g(x)$ in $\mathbb{F}_{q}[x]$ are called {multiplicative} equivalen{t} if there exists an integer $1 \le d \le q-1$ such that $\mathrm{gcd}(d,q-1)=1$ and $f(x)=g(x^d)$.
    \end{Def}
    {It should be noted that in  \cite{XH1} and many references before, two permutation polynomial  $f$ and $g$ of $\gf_q$ called equivalent if $f(x)=cg(ax+b)+d$, where
     $a, c\in \gf_q^*$, $b, d\in \gf_q$. In our opinion, this type of equivalence can be called {linear} equivalence.  }

{{The following lemmas will be useful in our future discussion.} The first one is a corollary following directly from } {\cite[Lemma 2.1]{Zieve}}.
      \begin{Lemma}
       	Let $f(x)=x^{r_{0}}(x^{(q-1)/d}+a)\in \mathbb{F}_{q}[x]$, where $d \mid q-1$, $a \in \mathbb{F}_{q}^{*}$. If $f(x)$ is a permutation polynomial over $\mathbb{F}_{q}$ and $\mathrm{gcd}(r_{0}+d,(q-1)/d)=1$, then $g(x)=x^{r_{0}+d}(x^{(q-1)/d}+a)$ is
       	also a permutation polynomial over $\mathbb{F}_{q}$. 	
       	\label{ky1}
       \end{Lemma}

       \begin{Lemma}
       	\emph{ \cite{LN} (Hermite's Criterion)}
       	Let $\mathbb{F}_{q}$ be of characteristic $p$. Then $f \in \mathbb{F}_{q}[x]$ is a permutation polynomial of $\mathbb{F}_{q}$ if and only if the following two conditions hold:
       	 \begin{enumerate}[(i)]
       	 	\item $f$ has exactly one root in $\mathbb{F}_{q}$;
       	 	\item for each integer $t$ with $1 \le t \le q-2$ and $t \not \equiv 0 \mod{p}$, the reduction of $f(x)^{t} \mod{x^{q}-x}$ has degree $\le$ $q-2$.
       	 \end{enumerate}
       \end{Lemma}

       \begin{Lemma}
       	\emph{(Lucas formula)}
       	Let $n$,$i$ be positive integers and $p$ be a prime. Assume $n={{a}_{m}}{{p}^{m}}+\cdots +{{a}_{1}}p+{{a}_{0}}$ and $i={{b}_{m}}{{p}^{m}}+\cdots +{{b}_{1}}p+{{b}_{0}}$. Then
       	$\left( \begin{matrix}
       	n  \\
       	i  \\
       	\end{matrix} \right)\equiv \left( \begin{matrix}
       	{{a}_{m}}  \\
       	{{b}_{m}}  \\
       	\end{matrix} \right)\left( \begin{matrix}
       	{{a}_{m-1}}  \\
       	{{b}_{m-1}}  \\
       	\end{matrix} \right)\cdots \left( \begin{matrix}
       	{{a}_{0}}  \\
       	{{b}_{0}}  \\
       	\end{matrix} \right)(\bmod p)
       	$.
       \end{Lemma}

       \begin{Lemma}
       	Let $q=2^k$ and $k>0$ be an integer. Then $f(x)=x+x^2+x^4$ is a permutation polynomial over $\mathbb{F}_{q}$ if and only if $k \not \equiv 0(\bmod 3)$.
       	\label{lem1}
       \end{Lemma}
{    \begin{proof}
This lemma can be easily proved. We omit it here.    \end{proof} }

       \begin{Lemma}
       	\emph{ \cite{SZM}}
       	\label{san1}
       	{ Let $a \in \mathbb{F}_{2^k}^\ast$. Then the cubic equation $x^3+x=a$ has a unique solution in $ \mathbb{F}_{2^k}$ if and only if ${\operatorname{Tr}_k}\left(a^{-1}\right) \neq {\operatorname{Tr}_k}(1)$.}
       \end{Lemma}

    \section{Permutation binomials}
    \label{sec2}
     {G}iven $f(x)=x^{m}+ax^{n}\in \mathbb{F}_{q}[x]$, where $0<n<m<q$ and $a \in \mathbb{F}_{q}^{*}$, there exist integers $r,t,d>0$ with $\mathrm{gcd}(t,q-1)=1$ and $d \mid q-1$ such that $f({{x}^{t}})\equiv {{x}^{r}}({{x}^{(q-1)/d}}+a)(\bmod {{x}^{q}}-x)$ \cite{XH1}. Therefore, w.l.o.g, we only consider {permutation binomials with} the form $x^{r}(x^{(q-1)/d}+a)$. Inspired by Hou and Lappano \cite{XHSD}, in this section, we consider  permutation binomials over finite fields with even extensions, {that is, over $\gf_{q^2}$}.

     \subsection{{Known permutation binomials over $\gf_{q^2}$}}
     \label{kpb}
     First, we review the known permutation binomials over finite fields with even extensions.
               \begin{Th}
               	\emph{ \cite{Zieve}}
               	Let $r,d>0$ be integers such that $d\mid q^2-1$ and $a \in \mathbb{F}_{q^2}^{*}$. Further assume that $\eta +\frac{a}{\eta }\in {{\mu }_{(q^2-1)/d}}$ for all $\eta \in {{\mu }_{2d}}$. Then $x^{r}(x^{(q^2-1)/d}+a)$ is a permutation polynomial of $\mathbb{F}_{q^2}$ if and only if $-a\notin {{\mu }_{d}}$, $\mathrm{gcd}(2d,2r+(q^2-1)/d) \le 2$.
               \end{Th}

               \begin{Th}
               	\emph{ \cite{Zieve1}}
               	Let $r$ and $d$ be positive integers, and let $\beta \in \mathbb{F}_{q^{2}}$ be such that $\beta^{q+1}=1$. Then $x^{r+d(q-1)}+\beta^{-1}x^{r}$ is a permutation polynomial of $\mathbb{F}_{q^2}$ if and only if all of the following hold:
               	\begin{enumerate}[(i)]
               		\item $\mathrm{gcd}(r,q-1)=1$.
               		\item $\mathrm{gcd}(r-d,q+1)=1$.
               		\item $(-\beta)^{(q+1)/\mathrm{gcd}(q+1,d)} \neq 1$.
               	\end{enumerate}
               \end{Th}
               \begin{Th}
               	\emph{ \cite{XHSD}}
               	Let $f=ax+x^{3q-2}\in \mathbb{F}_{q^2}[x]$, where $a \in \mathbb{F}_{q^2}^{*}$. Then $f$ is a permutation polynomial of $\mathbb{F}_{q^2}$ if one of the following occurs.
               	\begin{enumerate}[(i)]
               		\item $q=2^{2k+1}$ and $a^{\frac{q+1}{3}}$ is a primitive $3$rd root of unity.
               		\item $q=5$ and $a^2$ is a root of $(x+1)(x+2)(x-2)(x^2-x+1)$.
               		\item $q=2^3$ and $a^3$ is a root of $x^3+x+1$.
               		\item $q=11$ and $a^4$ is a root of $(x-5)(x+2)(x^2-x+1)$.
               		\item $q=17$ and $a^6=4,5$.
               		\item $q=23$ and $a^8=-1$.
               		\item $q=29$ and $a^{10}=-3$.
               	\end{enumerate}
               \end{Th}
               \begin{Th}
               \emph{ \cite{SD}}
               Let $f=ax+x^{5q-4}\in \mathbb{F}_{q^2}[x]$, where $a \in \mathbb{F}_{q^2}^{*}$. Then $f$ is a permutation polynomial of $\mathbb{F}_{q^2}$ if one of the following occurs.
               \begin{enumerate}[(i)]
               	\item $q=2^{4k+2}$ and $a^{\frac{q+1}{5}}$ is a primitive $5$th root of unity.
               	\item $q=3^2$ and $a^2$ is a root of $(x+1)(x^2+1)(x^2+x+2)(x^2+2x+2)(x^4+x^2+x+1)(x^4+x^3+x^2+1)(x^4+2x^3+x^2+2x+1)$.
               	
               	\item $q=19$ and $a^4$ is a root of $(x+1)(x+2)(x+3)(x+4)(x+5)(x+9)(x+10)(x+13)(x+17)(x^2+3x+16)(x^2+4x+1)(x^2+18x+6)$.
               	\item $q=29$ and $a^6\in \{15,18,22,23\}$.
               	\item $q=7^2$ and $a^{10}$ is a root of $x^2+4x+1$.
               	\item $q=59$ and $a^{12}$ is a root of $(x^2+x+1)(x^3+x+1)$.
               	\item $q=2^6$ and $a^{13}$ is a root of $(x^2+x+1)(x^3+x+1)$.
               \end{enumerate}
               \end{Th}

      \subsection{{N}ew classes of permutation binomials}

     \begin{Th}
     	Let $f(x)={{x}}({{x}^{q+1}}+a) \in \mathbb{F}_{q^{2}}[x]$. Then $f(x)$ is a permutation {binomial} over $\mathbb{F}_{q^{2}}$ if and only if $ q \not \equiv 1(\bmod3)$ and $a^{2(q-1)}-a^{q-1}+1=0$.
     	\label{th_BiPP1}
     \end{Th}
     \begin{proof}
     	First of all, we show that when $ q \equiv 1(\bmod3)$,  $f(x)$ is not a permutation polynomial over $\mathbb{F}_{q^{2}}$.
     	
     	Let $0\le \alpha ,\beta \le q-1$ with $(\alpha ,\beta) \ne (0,0)$. Then {w}e have
     	\begin{eqnarray*}
     		\sum\limits_{x\in {{\mathbb{F}}_{{{q}^{2}}}}}{f{{(x)}^{\alpha +q\beta }}}
     		&=& \sum\limits_{x\in {{\mathbb{F}}_{{{q}^{2}}}}}{{{(ax+{{x}^{q+2}})}^{\alpha }}{{({{a}^{q}}{{x}^{q}}+{{x}^{1+2q}})}^{\beta }}} \\
     		&=& \sum\limits_{x\in {{\mathbb{F}}_{{{q}^{2}}}}}{\sum\limits_{ 0\le i \le \alpha,  0\le j \le \beta}{\left( \begin{matrix}
     					\alpha   \\
     					i  \\
     				\end{matrix} \right){{a}^{\alpha -i}}{{x}^{\alpha -i}}{{x}^{(q+2)i}}\left( \begin{matrix}
     				\beta   \\
     				j  \\
     			\end{matrix} \right){{a}^{q(\beta -j)}}{{x}^{q(\beta -j)}}{{x}^{(1+2q)j}}}}\\
     	&=& {{a}^{\alpha +q\beta }}\sum\limits_{0\le i \le \alpha,  0\le j \le \beta}{\left( \begin{matrix}
     			\alpha   \\
     			i  \\
     		\end{matrix} \right)}\left( \begin{matrix}
     		\beta   \\
     		j  \\
     	\end{matrix} \right){{a}^{-i-qj}}\sum\limits_{x\in \mathbb{F}_{{{q}^{2}}}^{*}}{{{x}^{\alpha +q\beta +(1+q)(i+j)}}}.
     \end{eqnarray*}

     The inner sum is $0$ unless $\alpha +q\beta +(1+q)(i+j)\equiv 0 (\bmod q^2-1)$, {which can happen only if } $\alpha +\beta q\equiv 0 (\bmod q+1)$, or equivalently, $ \alpha=\beta$.  Then we have
     \begin{equation*}
     \sum\limits_{x\in \mathbb{F}_{{{q}^{2}}}^{*}}{f{{(x)}^{\alpha +q\beta }}=-{{a}^{(1+q)\alpha }}\sum\limits_{\tiny{\begin{array}{c}
                                                                                                                  \alpha +i+j\equiv 0(\bmod q-1), \\
                                                                                                                  0\le i, j \le \alpha
                                                                                                                \end{array}}
     }{\left( \begin{matrix}
     		\alpha   \\
     		i  \\
     		\end{matrix} \right)\left( \begin{matrix}
     		\alpha   \\
     		j  \\
     		\end{matrix} \right)}{{a}^{-i-qj}}}.
     \end{equation*}
     Let $\alpha =\frac{q-1}{3}$. Then
     \begin{equation*}
     \sum\limits_{x\in {\mathbb{F}_{{{q}^{2}}}}}{f{{(x)}^{\alpha +q\alpha }}=-}{{a}^{(q+1)\alpha }}{{a}^{-\alpha -\alpha q}}=-1\ne 0.
     \end{equation*}
     Hence, in the case, $f(x)$ is not a permutation polynomial over $\mathbb{F}_{q^{2}}$.

     In the following, we show that when $ q \not \equiv 1 (\bmod{3})$, $f(x)$ is a permutation polynomial over $\mathbb{F}_{q^{2}}$ if and only if $a^{2(q-1)}-a^{q-1}+1=0$.

     First, it is easy to prove that {zero is the only root of  $f(x)=0$ in $\gf_q$} if and only if ${{a}^{q-1}}\ne 1$.

     Next we prove that $f(x)=c$ has at most one solution in $\mathbb{F}_{q^{2}}$ for any $c \in \mathbb{F}_{q^{2}}^{*}$ {if and only if $a^{2(q-1)}-a^{q-1}+1=0$}. It is clear that $0$ is not a solution. Therefore we consider the following equation
     \begin{equation}
     {{x}^{q+1}}+a=\frac{c}{x}.
     \label{ky3}
     \end{equation}
     Let {$u=\frac{c}{x}-a=x^{q+1}$. Then $u\in \mathbb{F}_{q}^{*}$} and $x=\frac{c}{a+u}$. Plugging it into (\ref{ky3}), one have
     \begin{equation*}
     	\left(\frac{c}{a+u}\right)^{q+1}=u.
     \end{equation*}
     After simplifying and rearranging the terms, we get
     \begin{equation}
     {{u}^{3}}+({{a}^{q}}+a){{u}^{2}}+{{a}^{q+1}}u={{c}^{q+1}}.
     \label{ky4}
     \end{equation}
     Let ${{a }_{2}}={{a}^{q}}+a$, ${{a }_{1}}={{a}^{q+1}}$, ${{a }_{0}}=-{{c}^{q+1}}$. Then $f(x)=c$ has at most one solution in $\mathbb{F}_{q^{2}}$ if and only if $g(u)={{u}^{3}}+{{a }_{2}}{{u}^{2}}+{{a }_{1}}u$ is a permutation polynomial over $\mathbb{F}_{q}$. {In the following,
     we will use the table of normalized permutation polynomials over $\gf_q$ (Table 7.1 in  \cite{LN}). Recall that a polynomial
      $\phi$ is called in normalized form if $\phi$ is monic, $\phi(0)=0$, and when the degree $n$ of $\phi$
      is not divisible by the characteristic of $\gf_q$, the coefficient of $x^{n-1}$ is 0. The rest of the proof is split into two cases.}

     Case 1: $q\equiv 2(\bmod 3)$.

     Let $u=v-\frac{{{a }_{2}}}{3}$. Plugging it into (\ref{ky4}), we get
     $$v^3+\mu_{1}v+\mu_{2}=0.$$
     where ${{\mu }_{1}}=-\frac{1}{3}a _{2}^{2}+{{a }_{1}}$, ${{\mu }_{2}}=\frac{2}{27}a _{2}^{3}-\frac{1}{3}{{a }_{1}}{{a }_{2}}+{{a }_{0}}$. Then ${{\mu }_{1}},{{\mu }_{2}}\in \mathbb{F}_{q}$.

     As $f(x)=c$ has at most one solution in $\mathbb{F}_{q^{2}}$ if and only if $h(v)={{v}^{3}}+{{\mu }_{1}}v$ is a permutation polynomial over ${{\mathbb{F}}_{q}}$, {which holds if and only if ${{\mu }_{1}}=0$ by  \cite[Table 7.1]{LN}.} Recalling that ${{\mu }_{1}}=-\frac{1}{3}{{a}^{2}}({{a}^{2(q-1)}}-{{a}^{q-1}}+1)$, we come to the conclusion
      that $f(x)=c$ has at most one solution in $\mathbb{F}_{q^{2}}$ if and only if  ${{a}^{2(q-1)}}-{{a}^{q-1}}+1=0$.

     Case 2: $q\equiv 0 (\bmod 3)$.

     { In this case, $g(u)={{u}^{3}}+{{a }_{2}}{{u}^{2}}+{{a }_{1}}u$ is already a normalized polynomial. According to  \cite[Table 7.1]{LN},
       $x^{3}-ax$ ($a$ is not a square) and $x^3$ are the two  normalized permutation polynomials over $\mathbb{F}_{q}$ with degree three.
       On one hand, if $g$ permutes $\gf_q$, then we have ${{a }_{2}}=0$. Thus  ${a}^{q-1}+1=0$, or equivalently, $0=({a}^{q-1}+1)^2={{a}^{2(q-1)}}-{{a}^{q-1}}+1$.
        On the other hand, if ${a}^{q-1}+1=0$, then $g(u)=u^3-a^2u$. Clearly, $a^2$ is a non-square in $\mathbb{F}_{q}$ since $a \in \gf_{q^2}\setminus\mathbb{F}_{q}$. Hence $g(u)$ permutes  $\gf_q$.}


     We finish the proof.
    \end{proof}
 We briefly discuss the multiplicative inequivalence of the new
permutation binomial in Theorem \ref{th_BiPP1} with known ones. Due to the special and specific conditions in Theorem  3.5 and those in the last
subsection, the
permutation binomial $f(x)$ of Theorem \ref{th_BiPP1} is not multiplicative equivalent to any known permutation binomial.
{According to Lemma \ref{ky1},  if $\mathrm{gcd}(1+l(q-1),q+1)=1$, where $l$ is an integer, the polynomial $x^{1+l(q-1)}(x^{q+1}+a)$ is also a permutation binomial over $\mathbb{F}_{q^{2}}$.}


    Theorem \ref{th_BiPP1} considers permutation binomials of $\mathbb{F}_{q^{2}}$ with the form ${{x}}({{x}^{q+1}}+a)$.  Next, we consider another permutation binomial with the form   ${{x}^{r}}({{x}^{q-1}}+a)$,
    {where $r\in [1, q+1]$}.
    \begin{Th}
    \label{th_BiPP2}
    	Let $f(x)={{x}^{r}}({{x}^{q-1}}+a)\in \mathbb{F}_{{{q}^{2}}}[x]$, $r\in [1,q+1]$. Then $f(x)$ is a permutation
    {binomial} over $\mathbb{F}_{{{q}^{2}}}$ if and only if $r=1$ and ${{a}^{q+1}}\ne 1$.

    \end{Th}
    \begin{proof}
    	First, it is clear that $f(x)={{x}^{r}}({{x}^{q-1}}+a)$ has only one root in $\mathbb{F}_{{{q}^{2}}}$ if and only if ${{a}^{q+1}}\ne 1$.
    	Next, let $0\le \alpha ,\beta \le q-1$ with $(\alpha ,\beta) \ne (0,0),(q-1,q-1)$. We have
    	\begin{eqnarray*}
    		\sum\limits_{x\in {\mathbb{F}_{{{q}^{2}}}}}{f{{(x)}^{\alpha +\beta q}}} &=& \sum\limits_{x\in {\mathbb{F}_{{{q}^{2}}}}}{{{x}^{r(\alpha +\beta q)}}{{({{x}^{q-1}}+a)}^{(\alpha +\beta q)}}}\\
    		&=& \sum\limits_{x\in \mathbb{F}_{{{q}^{2}}}^{*}}{{{x}^{r(\alpha +\beta q)}}{{({{x}^{q-1}}+a)}^{\alpha }}{{({{x}^{1-q}}+{{a}^{q}})}^{\beta }}}\\
    		&=& \sum\limits_{x\in \mathbb{F}_{{{q}^{2}}}^{*}}{{{x}^{r(\alpha +\beta q)}}\sum\limits_{i=0}^{\alpha }{\left( \begin{matrix}
    					\alpha   \\
    					i  \\
   				\end{matrix} \right){{x}^{i(q-1)}}{{a}^{\alpha -i}}}\sum\limits_{j=0}^{\beta }{\left( \begin{matrix}
    				\beta   \\
    				j  \\
    			\end{matrix} \right){{x}^{j(1-q)}}{{(a^q)}^{\beta -j}}}}\\
    	&=& {{a}^{\alpha +\beta q}}\sum\limits_{0\le i\le \alpha,0 \le j \le \beta}{\left( \begin{matrix}
    			\alpha   \\
    			i  \\
    		\end{matrix} \right)}\left( \begin{matrix}
    		\beta   \\
    		j  \\
    	\end{matrix} \right){{a}^{-i-qj}}\sum\limits_{x\in \mathbb{F}_{{{q}^{2}}}^{*}}{{{x}^{r(\alpha +\beta q)+(i-j)(q-1)}}}.
    	\end{eqnarray*}

        The inner sum is 0 unless  $r(\alpha+\beta q)+(i-j)(q-1)\equiv 0(\bmod q^2-1) $,{which can happen only if } $\alpha +\beta q\equiv 0(\bmod q-1)$, or equivalently, $\alpha +\beta =q-1$. Let $\beta =q-1-\alpha $. Then we have
        \begin{eqnarray*}        	
        \sum\limits_{x\in {{\mathbb{F}}_{{{q}^{2}}}}}{f{{(x)}^{\alpha +(q-1-\alpha )q}}}  &=&	
        {{a}^{(\alpha +1)(1-q)}}\sum\limits_{0 \le i \le \alpha,0 \le j \le q-1-\alpha}{\left( \begin{matrix}
        		\alpha   \\
        		i  \\
        	\end{matrix} \right)}\left( \begin{matrix}
        	q-1-\alpha   \\
        	j  \\
        \end{matrix} \right){{a}^{-i-qj}}\sum\limits_{x\in \mathbb{F}_{{{q}^{2}}}^{*}}{{{x}^{(q-1)(-r(\alpha +1)+i-j)}}}  \\
       &=& -{{a}^{(\alpha +1)(1-q)}}\sum\limits_{-r(\alpha +1)+i-j\equiv 0 (\bmod q+1)}{\left( \begin{matrix}
        		\alpha   \\
        		i  \\
        	\end{matrix} \right)}\left( \begin{matrix}
        	q-1-\alpha   \\
        	j  \\
        \end{matrix} \right){{a}^{-i-qj}}.       	
    \end{eqnarray*}
    As $i$ runs over the interval $[0,\alpha]$ and $j$ over the interval $[0,q-1-\alpha]$, $-r(\alpha +1)+i-j\in {{S}_{\alpha ,r}}$, where
    $${{S}_{\alpha ,r}}=[-r(\alpha +1)-(q-1-\alpha ),-r(\alpha +1)+\alpha ].$$

    If $r=1$, then ${{S}_{\alpha ,r}}=[-q,-1]$, and the equation $-(\alpha +1)+i-j\equiv 0(\bmod q+1)$ has no solution. Therefore, $\sum\limits_{x\in {{\mathbb{F}}_{{{q}^{2}}}}}{f{{(x)}^{\alpha +(q-1-\alpha )q}}}=0$. Hence, $f(x)$ is a permutation polynomial over ${\mathbb{F}}_{{{q}^{2}}}$.

    If $r \in [2,q+1]$, let $\alpha=0$. Then ${{S}_{0,r}}=[-r-q+1,-r]$. In the case, the only multiple of $q+1$ in ${S}_{0,r}$ is $-(q+1)$. Therefore, we have
    \begin{equation*}
    \sum\limits_{x\in {{\mathbb{F}}_{{{q}^{2}}}}}{f{{(x)}^{(q-1)q}}}=-{{a}^{(1-q)}}\left( \begin{matrix}
    q-1  \\
    q+1-r  \\
    \end{matrix} \right){{a}^{-1-q+rq}}=-\left( \begin{matrix}
    q-1  \\
    q+1-r  \\
    \end{matrix} \right){{a}^{(r-2)q}}.
    \end{equation*}
    {It follows from Lucas formula that  $\left( \begin{matrix}
    q-1  \\
    q+1-r  \\
    \end{matrix} \right)\not\equiv 0 \bmod p
    $ if $r\in [2,q+1]$.} Therefore,  $\sum\limits_{x\in {{\mathbb{F}}_{{{q}^{2}}}}}{f{{(x)}^{(q-1)q}}}\ne 0 $.
    Hence, when $r \in [2,q+1]$, $f(x)$ is not a permutation polynomial over $\mathbb{F}_{q^{2}}$.

    The proof is finished.
    \end{proof}	

    {In Theorem \ref{th_BiPP2}, let $r=1$, then we get $f(x)={{x}^{r}}({{x}^{q-1}}+a)=x^q+ax$, which is a linearized polynomial.
    Hence Theorem \ref{th_BiPP2} does not provide new permutation binomial. However, to the authors' best
    knowledge,  the characterization of the sufficient and necessary condition for the polynomial with this form to
    be  a permutation is new.
     }


    \section{Permutation trinomials}
    \label{sec3}
    \subsection{The known permutation trinomials}
    \label{kpt}
    First, we review the known permutation trinomials over $\mathbb{F}_{2^m}$.
    The results before 2014 were summarized in  \cite{DQ}.  We copy their list into the following theorem for the readers' conveniences,

    \begin{Th}\label{th_TriPPList}
    \begin{enumerate}
\item Some linearized permutation trinomials described in \cite{LN}.

\item $x+x^3+x^5$ over $\gf_{2^m}$, where $m$ is odd (the Dickson
polynomial of degree 5).

\item $x+x^5+x^7$ over $\gf_{2^m}$, where $m \not\equiv 0  \pmod{3}$ (the Dickson
polynomial of degree 7).

\item $x+x^3+x^{2^{(m+1)/2}+1}$ over $\gf_{2^m}$, where $m$ is odd \cite{HD}.

\item $x^{2^{2k}+1} +(ax)^{2^k+1} + ax^2$ over $\gf_{2^m}$, where $m=3k$ and $a^{(2^m-1)/(2^k-1)} \ne 1$ \cite{BCHO}.

\item $x^{3 \cdot 2^{(m+1)/2}+4}+x^{2^{(m+1)/2}+2}+x^{2^{(m+1)/2}}$ over $\gf_{2^m}$,
where $m$ is odd (\cite{Ch} or \cite[Theorem 4]{HD1}).

\item   $x^{2^{2k}+1}+x^{2^k+1}+vx$ over $\gf_{2^m}$, where $m=3k$ and $v \in \gf_{2^k}^*$ \cite{TZH}.


\item  \cite{LeePark}\label{thm-LeePark} Let $q \equiv 1 \pmod{3}$ be a prime power. Let $\alpha$ be a generator of $\gf_{q}^*$,  $s=(q-1)/3$,
and let $\omega=\alpha^s$. Define $f(x)=ax^2+bx+c \in \gf_{q}[x]$. Then
$$
h(x):=x^rf(x^s)
$$
is a permutation polynomial over $\gf_q$ if and only if the following conditions are satisfied:
\begin{enumerate}
\item $\gcd(r, s)=1$,
\item $f(\omega^i) \ne 0$ for $0 \le i \le 2$,
\item $\log_{\alpha}(f(1)/f(\omega)) \equiv \log_{\alpha}(f(\omega)/f(\omega^2)) \not\equiv r \pmod{3}$.
\end{enumerate}

\end{enumerate}
\end{Th}

%
%
%

    The following results are new classes of permutation trinomials constructed in   \cite{DQ}.
         \begin{Th}\label{th_DQ1}
         	\emph{ \cite{DQ}}
         	Let $m>1$ be an odd integer. Then both $x+x^{2^{(m+1)/2}-1}+x^{2^{m}-2^{(m+1)/2}+1}$ and $x+x^{3}+x^{2^{m}-2^{(m+3)/2}+2}$ are permutation polynomials over $\mathbb{F}_{2^{m}}$.
         \end{Th}
         \begin{Th}
         	\emph{ \cite{DQ}}
         	Let $m$ be a positive even integer. Then $x+x^{2^{(m+2)/2}-1}+x^{2^{m}-2^{m/2}+1}$ is a permutation polynomial over $\mathbb{F}_{2^{m}}$.
         \end{Th}

         \begin{Th}
         	\emph{ \cite{DQ}}
         	Let $k$ be a positive integer and $q$ be a prime power with $q \not \equiv 0(\bmod 3)$. Let $m$ be a positive even integer. Then, $x+x^{kq^{m/2}-(k-1)}+x^{k+1-kq^{m/2}}$ is a permutation polynomial over  $\mathbb{F}_{q^{m}}$ if and only if one of the following three conditions holds:
         	\begin{enumerate}[(i)]
         		\item $m \equiv 0 (\bmod 4)$;
         		\item $q \equiv 1 (\bmod 4)$;
         		\item $m \equiv 2 (\bmod 4), q \equiv 2 (\bmod 3)$, and $\mathrm{exp}_{3}(q^{m/2}+1)$, where $\mathrm{exp}_{3}(i)$ denotes the exponent of $3$ in the canonical factorization of $i$.
         	\end{enumerate}
         \end{Th}

    Recently, some particular types of permutation trinomials have been determined by Hou. We list these following
    permutation trinomials in even characteristic.
    \begin{Th}
    \label{TPPHou1}
    	\emph{ \cite{XH4}}
         Let $f=ax+bx^q+x^{2q-1} \in \mathbb{F}_{q^2}[x]$, where $q$ is even. Then $f$ is a permutation polynomial over $\mathbb{F}_{q^2}$ if and only if one of the following is satisfied.
        \begin{enumerate}[(i)]
        	\item $a=b=0,q=2^{2k}$.
        	\item $ab\neq 0, a=b^{1-q}, \mathrm{Tr}_{q/2}(b^{-1-q})=0$.
        	\item $ab(a-b^{1-q})\neq 0, \frac{a}{b^2}\in \mathbb{F}_q,\mathrm{Tr}_{q/2}(\frac{a}{b^2})=0,b^2+a^2b^{q-1}+a=0$.
        \end{enumerate}
    \end{Th}
      \begin{Th}
         \label{TPPHou2}
      	\emph{ \cite{XH2}}
      	Let $q>2$ be even and $f=x+tx^q+x^{2q-1}\in \mathbb{F}_q[x]$, where $t\in \mathbb{F}_q^*$. Then $f$ is a permutation polynomial of $\mathbb{F}_{q^2}$ if and only if $\mathrm{Tr}_{q/2}(\frac{1}{t})=0$.
      \end{Th}


    \subsection{New Classes of Permutation Trinomials}
    {In this subsection, we introduce five new classes of permutation trinomials over $\mathbb{F}_{2^m}$.
    Motivated by \cite{DQ}, we first consider permutation trinomials with trivial coefficients, that is, whose nonzero coefficients are all $1$.
   }
     \begin{Th}
     	\label{theorem1}
     	Let $q=2^{2k}$ and $k$ be a positive integer. Then $f(x)=x+x^{2^{k}}+x^{2^{2k-1}-2^{k-1}+1}$  is a permutation trinomial over $\mathbb{F}_{q}$ if and only if $k \not \equiv 0(\bmod 3)$.
     \end{Th}
     \begin{proof}
     	First of all, we show that $x=0$ is the only solution of $f(x)=0$ in $\mathbb{F}_{q}$ when $k \not\equiv 0(\bmod 3)$. If $f(x)=0$, then either $x=0$ or $1+x^{2^{k}-1}+x^{2^{2k-1}-2^{k-1}}=0$. Therefore, we only need to prove {that} the equation
     	\begin{equation}
     	1+x^{2^{k}-1}+x^{2^{2k-1}-2^{k-1}}=0
     	\label{key18}    	
     	\end{equation}
     	has no solution in $\mathbb{F}_{q}$.

     Let $y=x^{2^{k}-1}$. Then computing $(y*(\ref{key18}))^2$ and simplifying it by $y^{2^k+1}=1$, we get
     	\begin{equation}
     	y+y^2+y^4=0.
     	\label{key20}
     	\end{equation}
{It follows from Lemma \ref{lem1} that (\ref{key20}) has no nonzero solution in $\mathbb{F}_{q}$ when
 $k \not\equiv 0(\bmod 3)$.
 Since $x=0$ is not the solution of (\ref{key18}), $f(x)=0$ has only one solution in $\mathbb{F}_{q}$
 when $k \not\equiv 0(\bmod 3)$.}
     	
     	Then we prove that $f(x)=a$ has at most one solution in $\mathbb{F}_{q}$ for any $a \in \mathbb{F}_{q}^{*}$ if and only if $k \not\equiv 0(\bmod 3)$.
     	Let $d=2^{2k-1}-2^{k-1}+1$.
     {Let $s$ be an integer satisfying $1\leq s \leq 2^{2k}-2$ and $4s \equiv 2^k+3 (\bmod{2^{2k}-1})$.
     Then it is easy to verify that $ds \equiv 1 (\bmod{2^{2k}-1})$.} And we have the following equation from $f(x)=a$:
     	\begin{equation}
     	x+x^{2^{k}}+x^{d}=a.
     	\label{key21}
     	\end{equation}
{Let $u=x^{d}+a=x+x^{2^k}$. Then $u\in \mathbb{F}_{2^{k}}$ and  $x=(a+u)^{s}$.} Plugging it into (\ref{key21}),we have
     	\begin{equation*}
     		(a+u)^s+(a+u)^{2^k\cdot s}=u.
     	\end{equation*}
     	Raising the above equation to the $4$-th power {and simplifying it}, we get
     	\begin{equation}
     	u^4+(a^{2^{k+1}}+a^2)u^2+(a^{2^{k}}+a)^{3}u+a^{2^{k}+1}(a^{2}+a^{2^{k+1}})=0.
     	\end{equation}
     	Let $b=a^{2^{k}}+a$, $c=a^{2^{k}+1}$. Then $b,c \in \mathbb{F}_{2^{k}}$, {and the above equation reduces to}
     	\begin{equation*}
     	u^{4}+b^{2}u^{2}+b^{3}u+b^{2}c=0.
     	\end{equation*}
     	
     	If $b=0$, then there is only one solution $u=0$ of the above equation. Hence,  $f(x)=a$ has at most one solution $x=a^{s}$ in $\mathbb{F}_{q}$ for any $a \in \mathbb{F}_{q}^{*}$.
     	
     	If $b \neq 0$, let $u=bv$. Rewriting the above equation, we get
   {  \begin{equation}\label{eqLv}
        v^{4}+v^{2}+v=\frac{c}{b^{2}}=\frac{a^{2^k+1}}{(a^{2^k}+a)^2}.
     \end{equation}

     	Since $x=(a+u)^{s}=(a+bv)^s$, it suffices to prove that \eqref{eqLv} has at most one solution in $\mathbb{F}_{2^k}$ for any $a \in \mathbb{F}_{q}^{*}$
     if and only if $k \not\equiv 0(\bmod 3)$. Then the result follows from Lemma \ref{lem1}.}

      We finish the proof.
%
     \end{proof}

    \begin{Th}
    	\label{theorem2}
    	Let $q=2^{2k}$ and $k>0$ be an odd integer. Then $f(x)=x+x^{2^{k}+2}+x^{2^{2k-1}+2^{k-1}+1}$ is a permutation trinomial over $\mathbb{F}_{q}$.
    \end{Th}
    \begin{proof}
    	First of all, we show that $x=0$ is the only solution of $f(x)=0$ in $\mathbb{F}_{q}$. If $f(x)=0$, then either $x=0$ or $1+{{x}^{{{2}^{k}}+1}}+{{x}^{{{2}^{2k-1}}+{{2}^{k-1}}}}=0$. Therefore, we only need to prove the equation
    	\begin{equation}
    	1+{{x}^{{{2}^{k}}+1}}+{{x}^{{{2}^{2k-1}}+{{2}^{k-1}}}}=0
    	\label{key12}    	
    	\end{equation}
    	has no solution in $\mathbb{F}_{q}$.
    	Raising (\ref{key12}) to {its} $2$-th power, we have
    	\begin{equation*}
    	1+x^{2^{k+1}+2}+x^{2^{k}+1}=0.
    	\end{equation*}
    	Let $y=x^{2^{k}+1}$. Then $y \in \mathbb{F}_{2^{k}}$, and we get
    	\begin{equation}
    	1+y+y^2=0.
    	\label{key13}
    	\end{equation}
    	It is clear that $y=1$ is not the solution of (\ref{key13}). Then multiplying $(1+y)$ to both sides, we have $y^3=1$. We know $\mathrm{gcd}(2^{k}-1,3)=1$ since $k$ is odd. Therefore the above equation has no solution in $\mathbb{F}_{2^{k}}$. Hence, $x=0$ is the only solution {of  $f(x)=0$ in $\gf_q$.}
    	
    		Next, we prove that $f(x)=a$ has at most one solution in $\mathbb{F}_{q}$ for any $a \in \mathbb{F}_{q}^{*}$.
    		Considering the following equation:
    		\begin{equation}
    		1+x^{2^{k}+1}+x^{2^{2k-1}+2^{k-1}}=\frac{a}{x}.
    		\label{key14}
    		\end{equation}
    		Let $u=\frac{a}{x}$. Then $u=1+x^{2^{k}+1}+x^{2^{2k-1}+2^{k-1}}\in \mathbb{F}^{*}_{2^k}$. Plugging $x=\frac{a}{u}$ into { the square of (\ref{key14}),}  we have
    		\begin{equation}
    		1+\frac{{{a}^{{{2}^{k+1}}+2}}}{{{u}^{4}}}+\frac{{{a}^{{{2}^{k}}+1}}}{{{u}^{2}}}={{u}^{2}}.
    		\label{key17}
    		\end{equation}
    		Let $y=\frac{a^{2^{k}+1}}{u^{2}}$. Then $y \in \mathbb{F}_{2^{k}}$. Plugging it into (\ref{key17}), we have $y^{3}+y^{2}+y=a^{2^{k}+1}$.
    		
    		{Since $g(y)=y^{3}+y^{2}+y=(y+1)^{3}+1$ is a permutation polynomial over $\mathbb{F}_{2^{k}}$ when $k>0$ is odd,} we know that $f(x)=a$ has at most one solution in $\mathbb{F}_q$ for any $a \in \mathbb{F}_{q}^{*}$. Hence the proof is complete.       		
  \end{proof}
%

      Next we introduce three new classes  of permutation trinomials with the form $f(x)=x+ax^{\alpha}+bx^{\beta}\in \gf_{2^m}[x]$,
     {where $a, b\in \gf_{2^m}^\ast$}.

      \begin{Th}
\label{th_Tripp3}
      	Let $q=2^{2k}$ and $k>0$ be an integer, $f(x)=x+ax^{2^{k+1}-1}+a^{2^{k-1}}x^{2^{2k}-2^{k}+1}$, where $a\in \mathbb{F}_{q}$ and the order of $a$ is $2^{k}+1$. Then $f(x)$  is a permutation trinomial over $\mathbb{F}_{q}$.
      \end{Th}
      \begin{proof}
      	First of all, we show that $x=0$ is the only solution to $f(x)=0$. If $f(x)=0$, then either $x=0$ or $1+ax^{2^{k+1}-2}+a^{2^{k-1}}x^{2^{2k}-2^{k}}=0$. Therefore, we only need to prove the equation
      	\begin{equation}
      	1+ax^{2^{k+1}-2}+a^{2^{k-1}}x^{2^{2k}-2^{k}}=0.
      	\label{key23}
      	\end{equation}
      has no solution in $\mathbb{F}_{q}$.
      	Adding $(\ref{key23})$ to $(\ref{key23})^{2^{k+1}}$, we have
      	\begin{equation*}
      	x^{3\cdot 2^{k}-3}=a^{3\cdot 2^{k-1}}.
      	\end{equation*}
      	Therefore, {$x^{2^{k}-1}=a^{2^{k-1}}$ or $x^{2^{k}-1}=a^{2^{k-1}}\omega$,} where $\omega^{2}+\omega+1=0$.
      	
      	If $x^{2^{k}-1}=a^{2^{k-1}}$, {then by plugging it into (\ref{key23}), one get $1+a^{2^k+1}+a^{2^{2k-1}+2^{k-1}}=0$.} Recalling the order of $a$ is $2^k+1$, we have $1=0$. It is a contradiction.
      	
      	If $x^{2^{k}-1}=a^{2^{k-1}}\omega$, {then by plugging it into (\ref{key23}), one have $\omega^{2^{k}}+\omega^{2}+1=0$. Hence  $\omega^{2^{k}}=\omega$, which means that $k$ is even. However, in this case, $(a^{2^{k-1}}\omega)^{2^k+1}=\omega^{2^k+1}=\omega^2\neq 1$. It then follows that } the equation $x^{2^{k}-1}=a^{2^{k-1}}\omega$ has no solution in { $\mathbb{F}_q$. Contradicts!}
      	
      	Hence, $x=0$ is the only solution to $f(x)=0$ {in $\gf_q$}.
      	
   	    Next we prove that $f(x)=c$ has at most one solution in $\mathbb{F}_q$ for any $c \in \mathbb{F}_{q}^{*}$. Considering the equation
   	    \begin{equation*}
   	    	x+ax^{2^{k+1}-1}+a^{2^{k-1}}x^{2-2^{k}}=c.
   	    \end{equation*}  	
      	Raising the above equation to the $2$-th power, we have
      	\begin{equation}
      	x^{2}+a^{2}x^{2^{k+2}-2}+a^{2^{k}}x^{4-2^{k+1}}=c^{2}.
      	\label{key25}
      	\end{equation}
      	Computing $(\ref{key25})+(\ref{key25})^{2^k}*a$, and simplifying it by using $a^{2^{k}+1}=1$, we have
      	\begin{equation*}
      		x^{2}+ax^{2^{k+1}}=c^{2}+ac^{2^{k+1}}.
      	\end{equation*}
      	Let $u=x^{2}+c^{2}$. Then $u=au^{2^{k}}$. Plugging $x^2=u+c^2$ into (\ref{key25}), we get
      {
            	\begin{equation}
      	u+a^2(u+c^2)^{2^{k+1}-1}+a^{2^k}(u+c^2)^{2-2^{k}}=0.
        \end{equation}
      Multiplying $(u+c^2)^{2^{k}+1}$ across both sides of the above equation and then substituting $u^{2^{k}}=u/a$ and $a^{2^k+1}=1$ into it,
      one can get the following equation after simplification. }
      	\begin{equation}
      	u^3+\alpha u+\beta=0,
      	\label{key27}
        \end{equation}
        where $\alpha=ac^{2^{k+1}+2}+c^{4}+a^{2}c^{2^{k+2}}$, $\beta=a^{3}c^{3\cdot 2^{k+1}}+c^{6}$.

       Now let $\gamma=ac^{2^{k+1}}+c^{2}$. Then $\beta=\gamma\alpha$. We distinguish two cases.

        {\bf Case  $\gamma=0$. }

        Then $\alpha=c^{4}$ and $\beta=0$. Therefore, the solutions of (\ref{key27}) are $u=0$ or $u=c^{2}$. So $x=0$ and $x=c$ may be the solutions of the equation $f(x)=c$. However, $f(0) \neq c$. Hence, in this case, $f(x)=c$ has  one unique nonzero solution $x=c$ for each nonzero $c \in \mathbb{F}_{q}$.

      {\bf Case  $\gamma\neq 0$. }

If $\alpha=0$, then $u=0$ is the only solution of (\ref{key27}).

  If $\alpha \neq 0$,     then let $u=\alpha^{1/2} \varepsilon$. Further, $ \alpha^{2^k}=\frac{\alpha}{a^2}$, $\varepsilon^{2^k}=\frac{u^{2^k}}{(\alpha^{1/2})^{2^k}}=\frac{u}{a}\frac{a}{\alpha^{1/2}}=\varepsilon$. Therefore $\varepsilon \in \mathbb{F}_{2^{k}}$.
      Dividing (\ref{key27}) by $\alpha^{\frac{3}{2}}$ results in
        \begin{equation*}
        	\varepsilon^{3}+\varepsilon+\frac{\gamma}{\alpha^{1/2}}=0.
        \end{equation*}
       It is routine to verify that $\frac{\gamma}{\alpha^{1/2}} \in \mathbb{F}_{2^{k}}$. Hence $h(\varepsilon)=	 \varepsilon^{3}+\varepsilon+\frac{\gamma}{\alpha^{1/2}} \in \mathbb{F}_{2^{k}}[\varepsilon]$.
      According to Lemma \ref{san1}, $h(\varepsilon)=0$ has only one solution in $\mathbb{F}_{2^{k}}$ if and only
      if ${{\operatorname{Tr}}_{k}}\left(\frac{\alpha }{{{\gamma }^{2}}}\right)={\operatorname{Tr}_k}(1)+1$.

     Next,  we show that ${{\operatorname{Tr}}_{k}}\left(\frac{\alpha }{{{\gamma }^{2}}}\right)={\operatorname{Tr}_k}(1)+1$. We know
        \begin{eqnarray*}
        {{\operatorname{Tr}}_{k}}\left(\frac{\alpha }{{{\gamma }^{2}}}\right)
        &=&      {{\operatorname{Tr}}_{k}}\left(\frac{a{{c}^{{{2}^{k+1}}+2}}+{{c}^{4}}+{{a}^{2}}{{c}^{{{2}^{k+2}}}}}{{{c}^{4}}+{{a}^{2}}{{c}^{{{2}^{k+2}}}}}\right)\\
        &=& {{\operatorname{Tr}}_{k}}\left(\frac{a{{c}^{{{2}^{k+1}}+2}}}{{{c}^{4}}+{{a}^{2}}{{c}^{{{2}^{k+2}}}}}+1\right)\\
        &=& {\operatorname{Tr}_k}(1)+{{\operatorname{Tr}}_{k}}\left(\frac{a{{c}^{{{2}^{k+1}}+2}}}{{{c}^{4}}+{{a}^{2}}{{c}^{{{2}^{k+2}}}}}\right)
        \end{eqnarray*}
        and
        \begin{equation*}
        {{\operatorname{Tr}}_{k}}\left(\frac{a{{c}^{{{2}^{k+1}}+2}}}{{{c}^{4}}+{{a}^{2}}{{c}^{{{2}^{k+2}}}}}\right)=
        {{\operatorname{Tr}}_{k}}\left(\frac{a{{c}^{{{2}^{k+1}}-2}}}{1+{{a}^{2}}{{c}^{{{2}^{k+2}}-4}}}\right)=
        {{\operatorname{Tr}}_{k}}\left(\frac{1}{1+a{{c}^{{{2}^{k+1}}-{2}}}}+\frac{1}{{{(1+a{{c}^{{{2}^{k+1}}-{2}}})}^{2}}}\right).
        \end{equation*}
        Let $v=\frac{1}{1+ac^{2^{{k+1}}-2}}$. {T}hen $v^{2^{k}}=\frac{1}{1+a^{2^{k}}c^{2-2^{2^{k+1}}}}=v+1$. Therefore we have
        \begin{equation*}
        {{\operatorname{Tr}}_{k}}\left(\frac{a{{c}^{{{2}^{k+1}}+2}}}{{{c}^{4}}+{{a}^{2}}{{c}^{{{2}^{k+2}}}}}\right)={{\operatorname{Tr}}_{k}}\left(v+{{v}^{2}}\right)=v+{{v}^{{{2}^{k}}}}=1.
        \end{equation*}
        Hence, ${{\operatorname{Tr}}_{k}}\left(\frac{\alpha }{{{\gamma }^{2}}}\right)={\operatorname{Tr}_k}(1)+1$.

       Then, ${\operatorname{Tr}_k}(\frac{\alpha^{1/2}}{\gamma})=\left({{\operatorname{Tr}}_{k}}\left(\frac{\alpha }{{{\gamma }^{2}}}\right)\right)^{\frac{1}{2}} \neq {\operatorname{Tr}_k}(1)$.  Hence,  $h(\varepsilon)=0$ has only one solution in $\mathbb{F}_{2^{k}}$ according to Lemma \ref{san1}.  Therefore $f(x)=c$ has at most one solution in $\mathbb{F}_q$ for any $c \in \mathbb{F}_{q}^{*}$.

        The proof is complete.
      \end{proof}

 \begin{Th}
\label{th_Tripp4}
 	Let $q=2^{2k+1}$ and $k>0$ be an integer. Then $f(x)=x+ax^{2^{k+1}-1}+a^{2^{2k+1}-2^{k+1}-2}x^{2^{k+1}+1}$, where $a \in \mathbb{F}_{q}$, is a permutation trinomial over $\mathbb{F}_{q}$.
 \end{Th}
 \begin{proof}
 	The case $a=0$ is trivial. Therefore, we assume $a\neq 0$ in the following.
 	
 	First of all, we show that $x=0$ is the only solution of $f(x)=0$ in $\mathbb{F}_q$. If $f(x)=0$, then either $x=0$ or $	 1+ax^{2^{k+1}-2}+a^{-2^{k+1}-1}x^{2^{k+1}}=0$. Therefore, we only need to prove the equation
 	\begin{equation}
 	1+ax^{2^{k+1}-2}+a^{-2^{k+1}-1}x^{2^{k+1}}=0
 	\label{key1}    	
 	\end{equation}
 	has no solution in $\mathbb{F}_{q}$.
 	Raising (\ref{key1}) to the $2^{k}$-th power, we have
 	\begin{equation}
 	1+a^{2^k}x^{1-2^{k+1}}+a^{-2^k-1}x=0.
 	\label{key2}
 	\end{equation}
 	Multiplying $a^{2^k+1}x^{2^{k+1}}$ {across}  both sides of (\ref{key2}), we have
 	\begin{equation}
 	(x+a^{2^k+1})^{2^{k+1}+1}=a^{2^{k+1}+2^{k}+2}.
 	\label{key3}
 	\end{equation}
 	{It is clear that $\mathrm{gcd}(2^{k+1}+1,2^{2k+1}-1)=1$. Hence $x=0$ is the only solution of (\ref{key3}) in $\gf_q$. }
 	However, it is not the solution of (\ref{key1}). Hence, (\ref{key1}) has no solution in $\mathbb{F}_{q}$.
 	
 	Then we prove that $f(x)=c$ has at most one solution in $\mathbb{F}_{q}$ for any $c \in \mathbb{F}_{q}^{*}$.
 	Considering the following equation:
 	\begin{equation}
 	x+ax^{2^{k+1}-1}+a^{-2^{k+1}-1}x^{2^{k+1}+1}+c=0.
 	\label{key4}
 	\end{equation}
 	Computing $a^{2^k}\ast(\ref{key4})+x^{2^k}\ast(\ref{key4})^{2^k}$, we have
 	\begin{equation}
 	a^{2^k+1}x^{2^{k+1}-1}+x^{2^{k+1}}+c^{2^k}x^{2^k}+a^{2^k}c=0.
 	\label{key9}
 	\end{equation}
 	Then {by computing $(c+x)x^{2^{k+1}-1}\ast(\ref{key9})^{2^{k+1}}+a^{2^{k+1}+1}x\ast(\ref{key4})$ and after simplification,} we get
 	\begin{equation}
 	(a^{2^{k+1}+2}+ac^{2^{k+1}}+c^2)x=ac^{2^{k+1}+1}.
 	\label{key11}
 	\end{equation}
 	{Since $a, c\in \gf_q^\ast$, we know that (\ref{key11}) has at most one solution in $\gf_q$.}
 	Then $f(x)=c$ has at most one solution in $\mathbb{F}_q$ for any $c \in \mathbb{F}_{q}^{*}$. {We are done.}   	
 \end{proof}
 {Let $a=1$ in the above theorem. Then $f(x)=x+x^{2^{k+1}-1}+x^{2^{k+1}+1}$ and
 $f(x^{2^{k+1}+2})=x^{2^{k+1}}+x^{2^{k+1}+2}+x^{3\cdot2^{k+1}+\bf{4}}$,
 which is the permutation trinomial in Theorem \ref{th_TriPPList}(6).
 Hence  Theorem  \ref{th_Tripp4} is a generalization of a known permutation trinomial.}

\mqu{A comment about Theorem \ref{th_Tripp4} is as follows. Just after our submission, we noticed  that this class of permutation trinomial was 
also introduced in a very recent paper. In \cite{MZFG}, Ma et al. presented several classes of new permutation polynomials,
including two classes of permutation trinomials, that is, \cite[Theorem 5.1 and Theorem 5.2]{MZFG}.
It can be readily verified that Theorem \ref{th_Tripp4} is equivalent to \cite[Theorem 5.1]{MZFG}, though our proof is different and a bit short.  }

 	The last class of permutation trinomial is a generalization of  the second permutation trinomial in Theorem \ref{th_DQ1} (\cite[Theorem 2.2]{DQ}).
      \begin{Th}
\label{th_Tripp5}
         	Let $q=2^{2k+1}$,$f(x)=x+ax^3+a^{2^{2k+1}-2^{k+1}}x^{2^{2k+1}-2^{k+2}+2}$, $a \in \mathbb{F}_{q}$. Then $f(x)$ is a permutation trinomial over $\mathbb{F}_{q}$.
         	\label{key29}
      \end{Th}

      \begin{proof}
      The case $a=0$ is trivial. If $a=1$, then it reduces to  \cite[Theorem 2.2]{DQ}. Therefore, we assume $a \neq 0, 1$ in the following.

      Let $y=x^{2^{k+1}}$ and $b=a^{2^{k+1}}$. Then $y^{2^{k+1}}=x^{2}$ and $b^{2^{k+1}}=a^{2}$. For all $x \in \mathbb{F}_{q}^{*}$, we have
      \begin{equation}
      f(x)=x+ax^3+\frac{a}{b}\frac{x^3}{y^2}=\frac{x(abx^2y^2+by^2+ax^2)}{by^2}.
      \label{key30}
      \end{equation}

      First, we show that $x=0$ is the only solution of $f(x)=0$ in $\mathbb{F}_q$. From (\ref{key30}), we only need to prove the equation
      	\begin{equation}
      	abx^2y^2+by^2+ax^2=0
      	\label{key31}
      	\end{equation}
      has no solution in {$\mathbb{F}_{q}^\ast$.}
      Raising (\ref{key31}) into {its} $2^{k+1}$-th power, we have
        \begin{equation}
        ba^{2}x^4y^2+a^{2}x^4+by^2=0.
        \label{key32}
        \end{equation}
        Computing $(\ref{key31})^{2}+(\ref{key32})$, we have
        \begin{equation*}
        (1+ax^{2})^{2^{k+1}+2}=0.
        \end{equation*}
        Then $x^{2}=\frac{1}{a}$ and $y^{2}=\frac{1}{b}$. Plugging them into (\ref{key31}), we get $1=0$. It is a contradiction.

        If $f(x)$ is not a permutation polynomial of $\mathbb{F}_{q}^{*}$, then  there exists $x \in \mathbb{F}_{q}^{*}$ and $c \in \mathbb{F}_{q}^{*}$ such that $f(x)=f((1+c)x)$. Let $d=c^{2^{k+1}}$. It is clear that $c,d\neq 0,1$ and $c\neq d$.  Then
        $$\frac{x(b{{y}^{2}}+ab{{x}^{2}}{{y}^{2}}+a{{x}^{2}})}{b{{y}^{2}}}=
        \frac{(1+c)x(b{{(1+d)}^{2}}{{y}^{2}}+ab{{(1+c)}^{2}}{{(1+d)}^{2}}{{x}
        ^{2}}{{y}^{2}}+a{{(1+c)}^{2}}{{x}^{2}})}{b{{(1+d)}^{2}}{{y}^{2}}}, $$

       After simplifying, we get
       \begin{equation}
        A_1x^2y^2+A_2y^2+A_3x^2=0,
        \label{key33}
       \end{equation}
       where
       \begin{eqnarray*}
  A_1 &=& abc(1+d)^2(c^2+c+1), \\
  A_2 &=& bc(1+d)^2, \\
  A_3 &=& a[(1+d)^2+(1+c)^3].
\end{eqnarray*}
      Raising (\ref{key33}) to $2^{k+1}$-th power, we have
      \begin{equation}
        A_1^{2^{k+1}}x^4y^2+A_3^{2^{k+1}}y^2 + A_2^{2^{k+1}}x^4=0.
        \label{key34}
       \end{equation}
   Computing (\ref{key33})*$(A_{1}^{2^{k+1}}x^{4}{+A_3^{2^{k+1}}})$+(\ref{key34})*$(A_{1}x^{2}+A_{2})$, we get
   \begin{equation}
   \label{key35}
  B_1x^4+B_2x^2 + B_3=0.
\end{equation}
where
\begin{eqnarray*}
  B_1 &=& A_3A_1^{2^{k+1}}+A_1A_2^{2^{k+1}}=a^3bd^2(1+c)^4[(1+c)^3+(1+d)^3], \\
  B_2 &=& A_2^{2^{k+1}+1}=a^2bcd(1+c)^4(1+d)^2, \\
 B_3 &=& A_3^{2^{k+1}+1}.
\end{eqnarray*}
 It is {routine to verify } that $A_1,A_2,B_1,B_2 \neq 0$. Let ${{x}^{2}}=\frac{{{B}_{2}}}{{{B}_{1}}}\varepsilon$. Plugging it into (\ref{key35}), we have
 \begin{equation}
 {{\varepsilon }^{2}}+\varepsilon +D=0,
 \label{key36}
 \end{equation}
 where $D=\frac{{{B}_{1}}{{B}_{3}}}{B_{2}^{2}}$.
 As for $D$, we know
 $$D=\frac{{{B}_{1}}{{B}_{3}}}{B_{2}^{2}}=\frac{A_{3}^{{{2}^{k+1}}+1}({{A}_{3}}A_{1}^{{{2}^{k+1}}}+{{A}_{1}}A_{2}^{{{2}^{k+1}}})}
 {A_{2}^{{{2}^{k+2}}+2}}=\frac{{{A}_{1}}A_{3}^{{{2}^{k+1}}+1}}{A_{2}^{{{2}^{k+1}}+2}}+\frac{A_{1}^{{{2}^{k+1}}}A_{3}^{{{2}^{k+1}}+2}}
 {A_{2}^{{{2}^{k+2}}+2}}={{D}_{1}}+D_{1}^{{{2}^{k+1}}},$$
where $${{D}_{1}}=\frac{{{A}_{1}}A_{3}^{{{2}^{k+1}}+1}}{A_{2}^{{{2}^{k+1}}+2}}=\frac{{{A}_{1}}{{B}_{3}}}{{{A}_{2}}{{B}_{2}}}=\frac{(c^2+c+1)[(1+c)^4+(1+d)^{3}][(1+d)^2+(1+c)^3]}{cd(1+d)^2(1+c)^4}.$$
{Now we claim that ${{\operatorname{Tr}}_{2k+1}}({{D}_{1}})=1$. It is the same claim appeared in the proof of  \cite[Theorem 2.2]{DQ}.
The proof of this claim is a bit long and intricate. We omit it here. The interested reader  please
refer  \cite{DQ} for details. }

Raising (\ref{key36}) to $2^{i}$-th power, where $i=0,1,\cdots,k $, and summing them, we get
$${{\varepsilon }^{{{2}^{k+1}}}}=\varepsilon +\sum\limits_{i=0}^{k}{{{({{D}_{1}}+D_{1}^{{{2}^{k+1}}})}^{{{2}^{i}}}}}={\varepsilon} +\sum\limits_{i=0}^{2k+1}{D_{1}^{{{2}^{i}}}}={\varepsilon}  +{{D}_{1}}+{{\operatorname{Tr}}_{2k+1}}({{D}_{1}})={\varepsilon}  +{{D}_{1}}+1.$$
and
$${{\varepsilon }^{{{2}^{k+1}}+1}}=\varepsilon (\varepsilon +{{D}_{1}}+1)={{D}_{1}}\varepsilon +D.$$
{Substituting ${{x}^{2}}=\frac{{{B}_{2}}}{{{B}_{1}}}\varepsilon$, ${{y}^{2}}=(\frac{{{B}_{2}}}{{{B}_{1}}})^{2^{k+1}}\varepsilon^{2^{k+1}}$ and} the above two equations into (\ref{key33}), we obtain
$$\frac{{{A}_{1}}B_{2}^{{{2}^{k+1}}+1}}{B_{1}^{{{2}^{k+1}}+1}}({{D}_{1}}\varepsilon +D)+\frac{{{A}_{2}}B_{2}^{{{2}^{k+1}}}}{B_{1}^{{{2}^{k+1}}}}(\varepsilon +{{D}_{1}}+1)+\frac{{{A}_{3}}{{B}_{2}}}{{{B}_{1}}}\varepsilon =0.$$
Multiplying $B_1^{2^{k+1}+1}$ across the two sides of the above equation and using $B_2^{2^{k+1}}=A_2^{2^{k+1}+2}=A_2B_2$, we have
$${{C}_{1}}\varepsilon +{{C}_{2}}=0,$$
where
\begin{eqnarray*}
  C_1 &=& A_1A_2B_2D_1+A_2^2B_1+A_3B_1^{2^{k+1}}=A_1A_2^{2^{k+1}+2}, \\
  C_2 &=& A_1A_2B_2D+A_2^2B_1(D_1+1)=A_2^2B_1.
\end{eqnarray*}
Therefore, $\varepsilon =\frac{{{C}_{2}}}{{{C}_{1}}}=\frac{{{B}_{1}}}{{{A}_{1}}A_{2}^{{{2}^{k+1}}}}=
\frac{{{A}_{2}}{{B}_{1}}}{{{A}_{1}}{{B}_{2}}}$. Plugging it into (\ref{key36}), we have
$${{B}_{1}}A_{2}^{2}+{{A}_{1}}{{A}_{2}}{{B}_{2}}=A_{1}^{2}{{B}_{3}}.$$
That is $A_{1}^{{{2}^{k+1}}}A_{2}^{2}=A_{1}^{2}A_{3}^{{{2}^{k+1}}}$. {Substituting the definitions of $A_1,A_2,A_3$ into the above equation,} we get
$$a^2bd(1+c)^4(d^2+d+1)b^2c^2(1+d)^4=a^2b^2c^2(1+d)^4(c^2+c+1)^2b[(1+c)^4+(1+d)^3].$$
Simplifying it, we have ${{(1+d)}^{3}}={{(1+c)}^{6}}$. {Since $\mathrm{gcd}(3,2^{2k+1}-1)=1$, one get $d=c^{2}$. After raising it to the $2^{k+1}$-th power, we obtain $c^2=d^2=d$, which means $d=0, 1$. It is a contradiction since $d\neq 0,1$.}

Hence, the proof is complete.
      \end{proof}

{   \subsection{Multiplicative  inequivalence of the new permutation trinomials with known ones}
In this subsection, we briefly discuss the multiplicative  inequivalence of the new permutation
trinomials with known ones.

First, the last three classes of permutation trinomials are not with trivial coefficients. As
far as the authors know, there exist only five such classes of permutation trinomials in $\gf_{2^m}$. They
are the functions in Theorem  \ref{th_TriPPList} (5), (7), (8), Theorems
\ref{TPPHou1} and \ref{TPPHou2}.
Any of our newly constructed permutation trinomial is multiplicative inequivalent to those in Theorem  \ref{th_TriPPList} (5), (7)
since the latter permutations are defined over $\gf_{2^{3k}}$.
The three conditions in Theorem \ref{th_TriPPList} (8) are further investigated in \cite{LeePark}. In general,
no simple conditions can make $h(x)$ into a permutation trinomial. Hence it is multiplicative inequivalent to  any of our newly constructed permutation trinomial. The permutations in Theorems \ref{th_Tripp4}, \ref{th_Tripp5} are multiplicative inequivalent to those in Theorems \ref{TPPHou1}, \ref{TPPHou2}
 since they are defined over odd dimension, while the latter permutations are defined over even dimension.
 It can be easily verified that the function in Theorem \ref{th_Tripp3} is multiplicative inequivalent to those in Theorems \ref{TPPHou1}, \ref{TPPHou2}.
Hence the permutation trinomials in Theorems \ref{th_Tripp3}, \ref{th_Tripp4} and \ref{th_Tripp5} are
indeed multiplicative inequivalent to known permutations. Further, they are pairwise multiplicative inequivalent.

\mqu{In \cite{MZFG}, two classes of permutation trinomials were introduced. The first one is equivalent to
Theorem \ref{th_Tripp4}. The second one is as follows. 
\begin{Th}\cite[Theorem 5.2]{MZFG}\label{th_MZF2}
Let $m>1$ be an odd integer such that $m=2k-1$. Then $f(x) = x + ux^{2^k-1} +
u^{2^k}
x^{2^m-2^{k+1}+2}$, $u\in  \gf_{2^m}$, is a permutation polynomial over $\gf_{2^m}$.
\end{Th}
We only need to check the multiplicative equivalence between the permutation trinomials in Theorems \ref{th_Tripp5} and \ref{th_MZF2}.
 It can be easily verified that they are indeed multiplicative inequivalent.
}

Now let us discuss the first two classes of permutation trinomials with trivial coefficients.
We need to show they are multiplicative inequivalent to all the known classes. To this end,
we used a Magma program
to confirm this conclusion for at least one small field $\gf_{2^m}$, where $4 \le m \le 10$. Consequently,
these two classes of permutation trinomials are also new. Further, these two classes are also multiplicative inequivalent to each other.
}

%

    \section{Conclusions}
    \label{sec4}

 {Permutation binomials and permutation trinomials over finite fields are both interesting and important
   in theory and in many applications. In this paper,
we present several new classes of permutation binomials and permutation trinomials. These functions extend
the list of known such permutations.  However, from computer experiments, we found that there should be
more classes of permutation binomials and permutation trinomials. A complete determination of
all permutation binomials or all permutation trinomials over finite fields seems to be out of reach for the time bing.
The constructed functions here lay a solid foundation for the further research.

At last, we would like to mention that the constructed functions have many applications. For instances,
they can be employed in linear codes \cite{CCJ} and cyclic codes  \cite{CD},
they can also be used to construct highly nonlinear functions such as bent and semi-bent functions.
For more details, please refer to the last paragraph in \cite{DQ} and the references therein.  }

\end{document}